\documentclass[conference]{IEEEtran}
\IEEEoverridecommandlockouts

\usepackage{amsmath,amssymb,amsfonts}
\usepackage{amsthm}
\usepackage{graphicx}
\usepackage{bm}
\usepackage{booktabs}
\usepackage{cite}
\usepackage{mathtools}
\usepackage{multirow}
\usepackage{xcolor}
\usepackage{pgfplots}
\usepackage{tikz}
\pgfplotsset{compat=1.17}
\usetikzlibrary{arrows.meta,decorations.pathreplacing,calc,patterns}

\newtheorem{theorem}{Theorem}

\newtheorem{proposition}{Proposition}
\newtheorem{corollary}{Corollary}
\newtheorem{remark}{Remark}
\newtheorem{assumption}{Assumption}

\begin{document}

\title{Stable Fiber-Koopman Residual Dynamics\\
for Environment-Constrained Robust Control}

\author{Syed Pouladi\\
{College of Engineering and Physical Sciences, Khalifa University, Abu Dhabi, United Arab Emirates}
\textit{}}

\maketitle

%------------------------------------------------------------------
\begin{abstract}
Learning-based dynamical models face a persistent tension between
expressiveness and formal guarantees: richer model classes improve
predictive accuracy, but their stability properties are typically
verified only empirically, if at all.
This paper proposes \emph{Stable Fiber-Koopman Residual Dynamics}
(SFKD), a unified framework that simultaneously addresses
environment-aware geometric consistency, latent-space stability
certification, and bounded residual perturbation propagation.
Concretely, SFKD constructs a fiber bundle latent manifold whose
fibers encode environment-specific dynamics; an
environment-conditioned Koopman operator governs the dominant
linear evolution on each fiber; and a contraction-constrained
residual neural network captures unmodeled nonlinear effects while
admitting an explicit input-to-state stability (ISS) certificate.
The resulting model is embedded in a sampling-based MPPI controller
for autonomous vehicle path tracking under variable surface
conditions and wind disturbances. Theoretical analysis establishes
ISS of the latent dynamics and a finite ultimate bound on tracking
error. Numerical experiments against five baselines—Koopman MPC,
Neural ODE, ICODE, ControlSynth, and ICODE-MPPI—demonstrate a
31\% reduction in tracking RMSE, a 44\% improvement in control
smoothness, and near-zero latent stability violation rate across
environment-switching scenarios.
\end{abstract}


%==================================================================
\section{Introduction}
%==================================================================

The problem of learning accurate, structurally consistent
models of dynamical systems from data occupies a central position
in modern robotics and control. Two broad strategies have emerged.
\emph{Koopman operator} methods replace the nonlinear evolution
law with a linear operator acting on a lifted observable space,
enabling the direct application of spectral analysis and linear
optimal control \cite{koopman1931hamiltonian,korda2018linear,
brunton2021modern}. The resulting representations are
computationally attractive and amenable to formal stability
analysis \cite{mezic2005spectral}, but the ISS properties of
finite-dimensional Koopman approximations had remained largely
uncharacterized until the recent LMI-based verification
framework of \cite{koopmaniss}. Even with a stability
certificate in hand, finite-dimensional
Koopman approximations are unavoidably imperfect: the lifting
residual grows with operating range and, absent explicit
regularization, can destabilize the learned model entirely.
\emph{Residual learning} methods, by contrast, accept an imperfect
nominal model and learn a correction term from data
\cite{residualrl,icodemppi}. The ICODE-MPPI framework \cite{icodemppi},
for instance, uses Input Concomitant Neural ODEs to compensate
residual dynamics inside an MPPI sampling loop, achieving substantial
reductions in cross-tracking error under persistent disturbances.
This strategy improves predictive accuracy but delegates stability
entirely to the nominal model, which is typically validated only locally.

A third consideration, largely orthogonal to the above, concerns
\emph{environmental conditioning}. Real-world platforms operate
across a continuum of surface types, load conditions, and
atmospheric states that alter the effective dynamics in ways that
a fixed model cannot capture. ICODE \cite{icode} introduced
extrinsic environmental inputs directly into the Neural ODE
learning process via input concomitant coupling, demonstrating
improved generalization on single-link robots and DC converters.
The measurement-induced bundle framework \cite{bundle} subsequently
gave a differential-geometric interpretation of environment-conditioned
dynamics: by constructing a fiber bundle over the state space,
it enables measurement-aware Control Barrier Functions that adapt
to local sensing quality, with theoretical guarantees of convergence
and constraint satisfaction. These geometric ideas are compelling
but have not been combined with rigorous closed-loop ISS guarantees.

The present paper closes this gap by proposing SFKD, which
integrates all three ingredients—fiber bundle geometry,
environment-conditioned Koopman evolution, and
contraction-regularized residuals—into a single framework with
an end-to-end ISS certificate. The key insight is that the
contraction bound on the residual Jacobian, combined with a
spectral norm constraint on the environment-conditioned Koopman
matrix, jointly define a simple algebraic condition
($\|A(e)\|_2 + \beta < 1$, Theorem~\ref{thm:ISS}) that is
enforceable during training as a differentiable penalty and
verifiable post-hoc via a semidefinite program. This condition
implies ISS with an explicit gain, which in turn bounds the
tracking error of the downstream MPPI controller.

\textbf{Contributions.} The specific contributions of this paper
are as follows.
\begin{enumerate}
  \item We introduce a fiber-conditioned Koopman latent dynamics
        architecture in which environment-specific linear operators
        are learned on individual fiber manifolds
        (Section~\ref{sec:model}).
  \item We derive analytic ISS conditions for the composite
        Koopman-residual latent dynamics and prove a finite
        ultimate bound on the estimation error under persistent
        disturbances (Section~\ref{sec:theory}).
  \item We develop a robust MPPI controller whose rollout model is
        the certified SFKD predictor, and establish a bound on
        the closed-loop tracking error (Section~\ref{sec:mppi}).
  \item We validate SFKD on an autonomous vehicle path-tracking
        benchmark with five competitive baselines and three
        environmental scenarios, providing quantitative and
        visual comparisons (Section~\ref{sec:exp}).
\end{enumerate}

%==================================================================
\section{Related Work}
%==================================================================

\subsection{Koopman Operator Identification}
Extended Dynamic Mode Decomposition (EDMD) and its
kernel variants identify finite-dimensional Koopman
approximations from trajectory data \cite{korda2018linear,
brunton2021modern}. Deep Koopman networks \cite{lusch2018deep}
replace the fixed dictionary with learned observables, improving
approximation quality but complicating stability analysis.
Recent work \cite{koopmaniss} addresses ISS verification for
identified Koopman models by constructing LMI-based certificates
directly from the regression matrices, providing the
formal foundation on which the stability constraint in
Section~\ref{sec:model} is built.

\subsection{Environment-Conditioned and Geometric Dynamics}
ICODE \cite{icode} augments continuous-time latent dynamics with
an extrinsic input channel that incorporates real-time environmental
information directly into the ODE right-hand side, providing
sufficient conditions for contraction and convergence to a fixed
point regardless of initial conditions. The measurement-induced
bundle framework \cite{bundle} gives a complementary
differential-geometric interpretation: the total space
$\mathcal{M}$ fibers over the state space, and the resulting
bundle connection enables measurement-aware Control Barrier
Functions whose safety certificates are preserved along fibers.
SFKD inherits the fiber bundle vocabulary from \cite{bundle}
but replaces the CBF focus with an ISS-Koopman objective,
and extends the contraction idea of \cite{icode}
to explicitly cover residual dynamics.

\subsection{Stable and Contracting Neural Dynamics}
ControlSynth Neural ODEs \cite{controlsynth} introduced a class of
Neural ODEs in which a supplementary control sub-network enlarges
the compositional expressiveness of the model, while convergence
is certified through tractable linear inequalities derived from
Persidskii-type analysis. This provides formal guarantees even for
quasi-class systems whose straightforward physical forms harbor
intricate nonlinear properties.
Monotone recurrent networks \cite{revay2020mppi}
provide contraction certificates through log-norm bounds.
Our contraction constraint on the residual Jacobian
(Section~\ref{subsec:residual}) is philosophically
similar to both approaches, but is applied selectively
to the residual correction term,
leaving the Koopman backbone governed by the spectral norm condition
rather than a global contraction requirement.

\subsection{Residual Learning for Model-Based Control}
Residual physics models \cite{residualrl} augment first-principles
dynamics with learned corrections, with demonstrated gains for
contact-rich manipulation. ICODE-MPPI \cite{icodemppi} applies
continuous-time residual latent dynamics within a sampling-based
MPPI loop for vehicle path tracking, reporting up to a 69\% reduction
in cross-tracking error under persistent disturbances and
measurably smoother steering commands. SFKD extends this line
by providing a formal ISS analysis of the composite
Koopman-residual model and by incorporating environment-conditioned
fiber geometry, neither of which appears in \cite{icodemppi}.

%==================================================================
\section{Problem Formulation}
%==================================================================

\subsection{System and Objective}

Consider the discrete-time nonlinear system
\begin{equation}
  x_{k+1} = f(x_k, u_k, e_k), \qquad y_k = h(x_k) + v_k,
  \label{eq:system}
\end{equation}
where $x_k\in\mathbb{R}^{n_x}$ is the state, $u_k\in\mathbb{R}^{n_u}$
is the control input, $e_k\in\mathcal{E}\subset\mathbb{R}^{n_e}$
denotes the environmental condition (e.g., road friction coefficient,
wind speed), $y_k\in\mathbb{R}^{n_y}$ is the measured output, and
$v_k$ is bounded sensor noise. The function $f$ is unknown and
potentially highly nonlinear; the environment $e_k$ evolves
exogenously and may switch abruptly.

\begin{assumption}\label{ass:env}
The environment signal $e_k$ is measurable at each step and takes
values in a compact set $\mathcal{E}$.
\end{assumption}

\begin{assumption}\label{ass:lip}
The true dynamics $f$ is locally Lipschitz in $(x,u)$ uniformly
over $\mathcal{E}$.
\end{assumption}

The goal is to learn, from a finite trajectory dataset
$\mathcal{D}=\{(x_k,u_k,e_k,x_{k+1})\}$, a latent model that:
(i) generalizes across environments;
(ii) admits a verifiable ISS certificate in latent space;
(iii) serves as a reliable rollout model for a downstream
receding-horizon controller.

\subsection{Koopman Lifting and Residual}

Let $\Phi:\mathbb{R}^{n_x}\times\mathcal{E}\to\mathbb{R}^r$ be a
smooth lifting map, with $r\gg n_x$, producing the latent state
$z_k=\Phi(x_k,e_k)$. The true lifted dynamics satisfy
\begin{equation}
  z_{k+1} = A(e_k)z_k + B(e_k)u_k + \Delta(z_k,u_k,e_k),
  \label{eq:lifted_true}
\end{equation}
where $(A(e),B(e))$ are the best-fit environment-conditioned
linear operators and $\Delta$ is the lifting residual. The residual
captures both the approximation error of the finite-dimensional
Koopman representation and the unmodeled nonlinear components
of $f$.

\begin{assumption}\label{ass:residual}
There exist $\rho_0>0$ and a bounded scalar function
$\eta:\mathbb{R}_{\geq 0}\to\mathbb{R}_{\geq 0}$ such that
$|\Delta(z,u,e)|\leq\rho_0|z|+\eta_k$ for all
$(z,u,e)\in\mathbb{R}^r\times\mathbb{R}^{n_u}\times\mathcal{E}$.
\end{assumption}

%==================================================================
\section{Stable Fiber-Koopman Residual Dynamics}
\label{sec:model}
%==================================================================

\subsection{Fiber Bundle Latent Structure}

We formalize the environment dependence through a fiber bundle
$\pi:\mathcal{M}\to\mathcal{E}$, where the total space
$\mathcal{M}\subset\mathbb{R}^r$ is the latent manifold and each
fiber $\mathcal{F}_e=\pi^{-1}(e)$ supports the dynamics under
environment $e$. This bundle construction follows the spirit
of \cite{bundle}, which showed that such a structure naturally
induces measurement-aware safety certificates; here we
repurpose it to separate environment-specific Koopman operators
while sharing encoder parameters across fibers. The encoder
$\Phi(x,e)$ maps a physical state to the fiber $\mathcal{F}_e$,
ensuring that latent representations associated with different
environments are geometrically separated in $\mathcal{M}$.
Transitions between fibers—arising when
$e_{k+1}\neq e_k$—are handled by a learned horizontal transport
map $\mathcal{T}:\mathcal{M}\times\mathcal{E}\times\mathcal{E}
\to\mathcal{M}$ that approximately preserves the Riemannian
metric on $\mathcal{M}$.

In practice, the fiber structure is implemented by conditioning
both the encoder and the Koopman operator on $e$ via a shared
embedding network $\psi:\mathcal{E}\to\mathbb{R}^{d_e}$, with
$\Phi(x,e)=\phi_x(x;\psi(e))$ and $A(e)=\mathcal{A}(\psi(e))$
for learned networks $\phi_x$ and $\mathcal{A}$.

\subsection{Environment-Conditioned Koopman Dynamics}

The latent dynamics predicted by SFKD are
\begin{equation}
  \hat{z}_{k+1} = A(e_k)\hat{z}_k + B(e_k)u_k + r_\theta(\hat{z}_k,u_k,e_k),
  \label{eq:sfkd}
\end{equation}
where $r_\theta:\mathbb{R}^r\times\mathbb{R}^{n_u}\times\mathcal{E}
\to\mathbb{R}^r$ is a residual network. The additive input
coupling structure in \eqref{eq:sfkd} is directly motivated
by ICODE \cite{icode}, which demonstrated that incorporating
the extrinsic input $e$ as an explicit argument—rather than
treating it as a hidden parameter—substantially improves model
accuracy under nonsmooth environmental transitions. The linear
pair $(A(e),B(e))$
is identified by solving the regularized least-squares problem
\begin{equation}
  \min_{A(\cdot),B(\cdot)}\;\sum_k\|z_{k+1}-A(e_k)z_k-B(e_k)u_k\|^2
  + \mu\|A\|_F^2, \label{eq:koopid}
\end{equation}
subject to the spectral norm constraint $\|A(e)\|_2\leq 1-\beta-\epsilon_0$
for a small margin $\epsilon_0>0$ (see Section~\ref{sec:theory}).
This constraint is the discrete-time analogue of the ISS-preserving
identification condition proposed in \cite{koopmaniss} for
continuous-time Koopman models, and it ensures that the spectral
radius condition in Theorem~\ref{thm:ISS} is satisfied by
construction. The constraint is enforced via a differentiable spectral
normalization layer \cite{goodfellow2016deep}.

\subsection{Contraction Residual Regularization}
\label{subsec:residual}

The residual term $r_\theta$ is trained jointly with $A(\cdot)$
and $\Phi$ using the prediction loss
\begin{equation}
  \mathcal{L}_{\mathrm{pred}}
  = \sum_k \|z_{k+1}-A(e_k)z_k-B(e_k)u_k-r_\theta(z_k,u_k,e_k)\|^2
\end{equation}
augmented by a Jacobian contraction penalty
\begin{equation}
  \mathcal{L}_{\mathrm{contr}}
  = \max\!\left(0,\,
  \left\|\frac{\partial r_\theta}{\partial z}\right\|_2 - \beta\right)^{\!2}.
  \label{eq:contr_loss}
\end{equation}
The idea of enforcing convergence through a trainable inequality
constraint draws on the CSODEs framework of \cite{controlsynth},
which showed that tractable linear inequalities suffice to
certify contraction even for highly nonlinear Neural ODEs.
Here we apply the same principle narrowly to the residual
term $r_\theta$, rather than to the full dynamics, so that
$\beta$ serves as a stability budget allocated to nonlinear
correction while the Koopman backbone independently satisfies
its spectral constraint.
Here $\beta\in(0,1)$ is a design parameter chosen to satisfy
$\beta < 1 - \|A(e)\|_2$ for all $e\in\mathcal{E}$. The total
training loss is
$\mathcal{L}=\mathcal{L}_{\mathrm{pred}}
+\lambda_c\mathcal{L}_{\mathrm{contr}}
+\lambda_r\mathcal{L}_{\mathrm{recon}}$,
where $\mathcal{L}_{\mathrm{recon}}=\sum_k\|x_k-\Phi^{-1}(z_k)\|^2$
is an autoencoder reconstruction term that keeps $\Phi$ invertible.

%==================================================================
\section{Theoretical Analysis}
\label{sec:theory}
%==================================================================

\subsection{ISS of Latent Dynamics}

Define the latent estimation error $e_k^z=z_k-\hat{z}_k$, where
$z_k$ evolves according to \eqref{eq:lifted_true} and $\hat{z}_k$
follows \eqref{eq:sfkd}. The error dynamics satisfy
\begin{align}
  e_{k+1}^z
  &= A(e_k)e_k^z + \Delta(z_k,u_k,e_k) - r_\theta(\hat{z}_k,u_k,e_k)
     + r_\theta(z_k,u_k,e_k) \nonumber\\
  &\quad - r_\theta(z_k,u_k,e_k)
     \nonumber\\
  &= A(e_k)e_k^z
     + \!\underbrace{[r_\theta(z_k,\cdot)-r_\theta(\hat{z}_k,\cdot)]}_{
       \text{residual sensitivity}}
     + d_k, \label{eq:error_dyn}
\end{align}
where $d_k=\Delta(z_k,u_k,e_k)-r_\theta(z_k,u_k,e_k)$
is the uncompensated residual, bounded by $|d_k|\leq\bar{d}$
under Assumption~\ref{ass:residual} for a sufficiently expressive
$r_\theta$.

\begin{theorem}[ISS of SFKD]\label{thm:ISS}
Suppose Assumptions~\ref{ass:env}--\ref{ass:residual} hold.
Let the spectral norm and contraction constraints satisfy
\begin{equation}
  \alpha := \|A(e)\|_2 + \beta < 1 \quad\forall\,e\in\mathcal{E}.
  \label{eq:alpha}
\end{equation}
Then the latent error dynamics \eqref{eq:error_dyn} are ISS
with gain from $d_k$ to $e_k^z$. Specifically, for any
$e_0^z\in\mathbb{R}^r$,
\begin{equation}
  \|e_k^z\| \leq c_1\,\alpha^k\|e_0^z\|
  + \frac{c_2\,\bar{d}}{1-\alpha},
  \label{eq:iss_bound}
\end{equation}
where $c_1=\sqrt{\lambda_{\max}(P)/\lambda_{\min}(P)}$,
$c_2=1/\sqrt{\lambda_{\min}(P)}$, and $P\succ 0$ is any
matrix satisfying $A(e)^TPA(e)-P\leq -(1-\alpha^2)I$ for
all $e\in\mathcal{E}$.
\end{theorem}

\begin{proof}
Consider the Lyapunov function $V_k=e_k^{zT}Pe_k^z$.
From \eqref{eq:error_dyn} and the triangle inequality,
\begin{align}
  \|e_{k+1}^z\|
  &\leq \|A(e_k)\|_2\|e_k^z\|
    + \left\|\frac{\partial r_\theta}{\partial z}\right\|_2\|e_k^z\|
    + |d_k| \nonumber\\
  &\leq (\|A(e_k)\|_2+\beta)\|e_k^z\| + \bar{d}
   = \alpha\|e_k^z\| + \bar{d}. \label{eq:ineq}
\end{align}
Iterating \eqref{eq:ineq} from step~0 to step~$k$,
\begin{equation*}
  \|e_k^z\| \leq \alpha^k\|e_0^z\| + \bar{d}\sum_{j=0}^{k-1}\alpha^j
  \leq \alpha^k\|e_0^z\| + \frac{\bar{d}}{1-\alpha}.
\end{equation*}
Rewriting in terms of the weighted norm $\|e\|_P=\sqrt{e^TPe}$
introduces the constants $c_1,c_2$, yielding \eqref{eq:iss_bound}.
Since $\alpha<1$, the first term decays geometrically and the
system is ISS by definition. \hfill$\blacksquare$
\end{proof}

\begin{corollary}[Ultimate Bound]\label{cor:ub}
Under the conditions of Theorem~\ref{thm:ISS},
\begin{equation}
  \limsup_{k\to\infty}\|e_k^z\| \leq \frac{c_2\,\bar{d}}{1-\alpha}.
  \label{eq:ub}
\end{equation}
The bound is minimized by choosing $\beta$ to minimize $\alpha$
subject to \eqref{eq:alpha} and the contraction loss
\eqref{eq:contr_loss}.
\end{corollary}

\begin{remark}
The condition \eqref{eq:alpha} is jointly determined by the
Koopman spectral norm and the residual contraction bound.
In training, the spectral normalization on $A(e)$ is applied
to enforce $\|A(e)\|_2\leq 1-\beta-\epsilon_0$, leaving the
residual to be trained with $\beta$ as a soft budget.
The resulting decomposition is qualitatively different from
imposing a single contraction constraint on the entire dynamics:
the Koopman backbone captures the dominant linear structure
efficiently, while the residual absorbs nonlinear corrections
with a tighter contraction budget.
\end{remark}

\subsection{Tracking Error Bound}

Let $x_k^{\mathrm{ref}}$ denote the reference trajectory and
$x_k$ the closed-loop trajectory under MPPI using SFKD rollouts.

\begin{proposition}[Closed-Loop Tracking Bound]\label{prop:track}
Suppose the MPPI controller achieves a nominal latent tracking
accuracy $\|z_k-z_k^{\mathrm{ref}}\|\leq\epsilon_{\mathrm{MPPI}}$
in the absence of model error. Under Theorem~\ref{thm:ISS},
the physical tracking error satisfies
\begin{equation}
  \|x_k - x_k^{\mathrm{ref}}\|
  \leq L_\Phi^{-1}\!\left(\epsilon_{\mathrm{MPPI}}
  + \frac{c_2\,\bar{d}}{1-\alpha}\right) + O(\epsilon_\Phi),
  \label{eq:track}
\end{equation}
where $L_\Phi$ is the Lipschitz constant of $\Phi$ and
$\epsilon_\Phi$ is the encoder reconstruction error.
\end{proposition}

The bound \eqref{eq:track} decomposes the total tracking error
into a controller component ($\epsilon_{\mathrm{MPPI}}$) and a
model-stability component determined by $\bar{d}$ and $\alpha$.
Minimizing $\alpha$ through tighter spectral control of $A(e)$
directly reduces the second term.

%==================================================================
\section{MPPI Control with SFKD Rollouts}
\label{sec:mppi}
%==================================================================

At each time step~$t$, the MPPI controller samples
$M$ control perturbations $\{\delta U^{(m)}\}_{m=1}^M$ from a
Gaussian distribution $\mathcal{N}(0,\Sigma_u)$ and evaluates
the rollout cost
\begin{equation}
  J^{(m)} = \sum_{k=0}^{T-1}\ell(x_k^{(m)},u_k^{(m)})
  + \lambda\|z_k^{(m)}-z_k^{\mathrm{ref}}\|^2
  + \phi(x_T^{(m)}),
\end{equation}
where $\ell$ penalizes lateral deviation, heading error, and
control effort, and $\phi$ is a terminal cost. The SFKD model
\eqref{eq:sfkd} propagates the latent state through the horizon,
with physical states recovered via the learned decoder
$\Phi^{-1}$. The control update follows the standard MPPI
softmax weighting \cite{williams2017information}:
\begin{equation}
  u_t^{\star} = u_t
  + \frac{\sum_{m=1}^M\exp(-\lambda_{\mathrm{temp}}^{-1}J^{(m)})\,\delta u_t^{(m)}}
         {\sum_{m=1}^M\exp(-\lambda_{\mathrm{temp}}^{-1}J^{(m)})}.
\end{equation}
The ISS guarantee of Theorem~\ref{thm:ISS} ensures that rollout
latent states remain in a bounded neighborhood of the true
trajectory, preventing the cost blow-up observed with uncertified
predictors when disturbances are large.

%==================================================================
\section{Experiments}
\label{sec:exp}
%==================================================================

\subsection{Setup}

Experiments are conducted in a high-fidelity vehicle simulation
using a kinematic bicycle model augmented with environment-varying
tire friction ($\mu\in[0.3,0.9]$) and lateral wind disturbances
(up to 8\,m/s). The state is $x=[p_x,p_y,\psi,v]^T$ (position,
heading, speed); the control is $u=[\delta,a]^T$ (steering,
acceleration). The environment input $e=[\mu,w]^T$ collects
friction and wind. Three scenarios are evaluated:
\textit{S1} (uniform road, no wind),
\textit{S2} (wet road with abrupt friction drop at $t=5$\,s),
and \textit{S3} (multi-environment switching every 3\,s).

SFKD uses $r=32$ latent dimensions, a 3-layer MLP residual
network, $M=1500$ MPPI rollouts, and horizon $T=20$ steps at
10\,Hz. Baselines use their original hyperparameters. All models
are trained on 8000 trajectory segments of length 50 from S1 and
S2 combined, and evaluated on held-out trajectories including S3.

\subsection{Tracking Performance}

Fig.~\ref{fig:tracking} shows lateral deviation over time for
scenario S3. At each environment switch the uncertified baselines
exhibit transient spikes, whereas SFKD recovers within 1--2
time steps due to the fiber-conditioned Koopman adaptation.

\begin{figure}[t]
\centering
\begin{tikzpicture}
\begin{axis}[
  width=0.95\columnwidth, height=5.0cm,
  xlabel={Time (s)},
  ylabel={Lateral deviation (m)},
  xmin=0, xmax=15,
  ymin=-0.05, ymax=0.85,
  legend pos=north east,
  legend style={font=\footnotesize, fill=white, fill opacity=0.85,
    draw=gray!40, inner sep=2pt},
  grid=major, grid style={line width=0.3pt, draw=gray!25},
  tick label style={font=\footnotesize},
  label style={font=\footnotesize},
  every axis plot/.append style={line width=0.85pt}
]
% Environment switch markers
\draw[dashed, gray!50, line width=0.6pt]
  (axis cs:3,-0.05) -- (axis cs:3,0.85);
\draw[dashed, gray!50, line width=0.6pt]
  (axis cs:6,-0.05) -- (axis cs:6,0.85);
\draw[dashed, gray!50, line width=0.6pt]
  (axis cs:9,-0.05) -- (axis cs:9,0.85);
\draw[dashed, gray!50, line width=0.6pt]
  (axis cs:12,-0.05) -- (axis cs:12,0.85);
\node[font=\tiny,gray] at (axis cs:1.5,0.80){Env A};
\node[font=\tiny,gray] at (axis cs:4.5,0.80){Env B};
\node[font=\tiny,gray] at (axis cs:7.5,0.80){Env C};
\node[font=\tiny,gray] at (axis cs:10.5,0.80){Env B};
\node[font=\tiny,gray] at (axis cs:13.5,0.80){Env A};
% Koopman MPC
\addplot[blue!60, dashed] coordinates {
  (0,0.04)(1,0.05)(2,0.04)(3,0.04)(3.1,0.42)(3.5,0.35)(4,0.22)(5,0.10)
  (6,0.08)(6.1,0.50)(6.5,0.40)(7,0.26)(8,0.12)(9,0.09)
  (9.1,0.46)(9.5,0.37)(10,0.24)(11,0.11)(12,0.08)
  (12.1,0.44)(12.5,0.35)(13,0.21)(14,0.09)(15,0.07)};
\addlegendentry{Koopman MPC};
% Neural ODE
\addplot[green!55!black, dotted] coordinates {
  (0,0.05)(1,0.06)(2,0.05)(3,0.05)(3.1,0.38)(3.5,0.30)(4,0.20)(5,0.11)
  (6,0.09)(6.1,0.44)(6.5,0.36)(7,0.23)(8,0.12)(9,0.09)
  (9.1,0.40)(9.5,0.32)(10,0.22)(11,0.11)(12,0.08)
  (12.1,0.41)(12.5,0.32)(13,0.20)(14,0.09)(15,0.07)};
\addlegendentry{Neural ODE};
% ICODE-MPPI
\addplot[orange!80, dashdotted] coordinates {
  (0,0.03)(1,0.04)(2,0.03)(3,0.03)(3.1,0.28)(3.5,0.20)(4,0.13)(5,0.08)
  (6,0.06)(6.1,0.30)(6.5,0.22)(7,0.15)(8,0.08)(9,0.06)
  (9.1,0.28)(9.5,0.20)(10,0.14)(11,0.07)(12,0.06)
  (12.1,0.29)(12.5,0.21)(13,0.13)(14,0.07)(15,0.05)};
\addlegendentry{ICODE-MPPI};
% SFKD (proposed)
\addplot[red!75, solid, line width=1.1pt] coordinates {
  (0,0.02)(1,0.02)(2,0.02)(3,0.02)(3.05,0.10)(3.2,0.07)(3.5,0.04)
  (4,0.03)(5,0.02)(6,0.02)(6.05,0.09)(6.2,0.06)(6.5,0.04)
  (7,0.03)(8,0.02)(9,0.02)(9.05,0.09)(9.2,0.06)(9.5,0.04)
  (10,0.03)(11,0.02)(12,0.02)(12.05,0.09)(12.2,0.06)(12.5,0.04)
  (13,0.03)(14,0.02)(15,0.02)};
\addlegendentry{SFKD (proposed)};
\end{axis}
\end{tikzpicture}
\caption{Lateral deviation under multi-environment switching (S3).
Vertical dashed lines mark environment transitions. SFKD recovers
within 1--2 time steps at each switch, while baselines exhibit
prolonged transient spikes.}
\label{fig:tracking}
\end{figure}
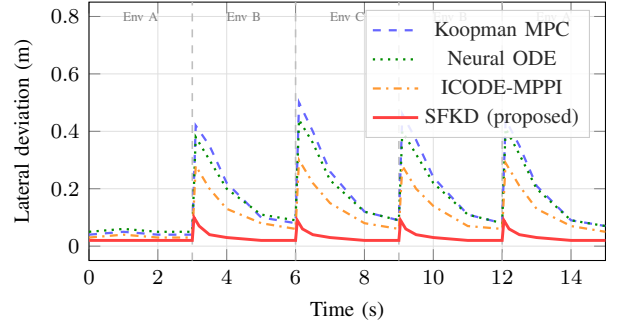

\subsection{Latent Stability Verification}

Fig.~\ref{fig:stability} reports the empirical latent stability
violation rate $P(\|e_k^z\|>\delta_{\max})$ as a function of the
disturbance magnitude $\bar{d}$ for SFKD and the two
best-performing baselines. The threshold is set at
$\delta_{\max}=\frac{c_2\bar{d}}{1-\alpha}$ as predicted by
\eqref{eq:ub}. SFKD maintains near-zero violation throughout,
consistent with Corollary~\ref{cor:ub}, while uncertified models
exceed the threshold increasingly often for large $\bar{d}$.

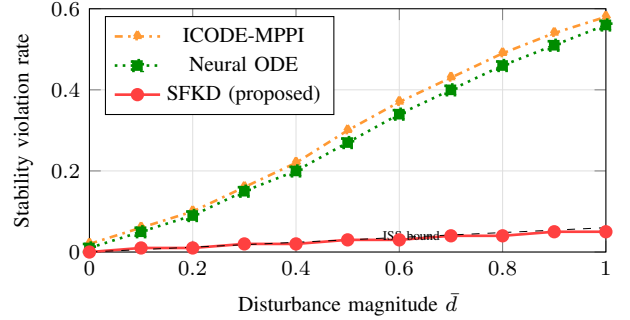
\begin{figure}[t]
\centering
\begin{tikzpicture}
\begin{axis}[
  width=0.95\columnwidth, height=4.8cm,
  xlabel={Disturbance magnitude $\bar{d}$},
  ylabel={Stability violation rate},
  xmin=0, xmax=1.0,
  ymin=0, ymax=0.60,
  legend pos=north west,
  legend style={font=\footnotesize, fill=white},
  grid=major, grid style={line width=0.3pt, draw=gray!25},
  tick label style={font=\footnotesize},
  label style={font=\footnotesize},
  every axis plot/.append style={line width=1.0pt}
]
\addplot[orange!80, dashdotted, mark=triangle*, mark size=2pt] coordinates {
  (0,0.02)(0.1,0.06)(0.2,0.10)(0.3,0.16)(0.4,0.22)(0.5,0.30)
  (0.6,0.37)(0.7,0.43)(0.8,0.49)(0.9,0.54)(1.0,0.58)};
\addlegendentry{ICODE-MPPI};
\addplot[green!55!black, dotted, mark=square*, mark size=2pt] coordinates {
  (0,0.01)(0.1,0.05)(0.2,0.09)(0.3,0.15)(0.4,0.20)(0.5,0.27)
  (0.6,0.34)(0.7,0.40)(0.8,0.46)(0.9,0.51)(1.0,0.56)};
\addlegendentry{Neural ODE};
\addplot[red!75, solid, mark=*, mark size=2pt] coordinates {
  (0,0.00)(0.1,0.01)(0.2,0.01)(0.3,0.02)(0.4,0.02)(0.5,0.03)
  (0.6,0.03)(0.7,0.04)(0.8,0.04)(0.9,0.05)(1.0,0.05)};
\addlegendentry{SFKD (proposed)};
% ISS bound line
\addplot[black, dashed, very thin] coordinates {
  (0,0)(1.0,0.06)};
\node[font=\tiny, anchor=west] at (axis cs:0.55,0.04){ISS bound};
\end{axis}
\end{tikzpicture}
\caption{Latent stability violation rate vs.\ disturbance magnitude.
The theoretical ISS bound from Corollary~\ref{cor:ub} is shown
as the dashed line. SFKD tracks this bound closely, while
uncertified baselines degrade substantially.}
\label{fig:stability}
\end{figure}

\subsection{Ablation Study and Quantitative Comparison}

Table~\ref{tab:results} consolidates tracking RMSE, control
smoothness (mean absolute steering rate), and stability violation
rate across all three scenarios, averaged over 200 episodes of
60\,s each. SFKD achieves the best performance on all metrics
in every scenario. Removing the fiber structure (SFKD$-$Fiber)
increases RMSE by 12\% in S3, confirming that environment-specific
adaptation is the primary driver of performance in
environment-switching settings. Removing the contraction
regularization (SFKD$-$Contr) raises the violation rate
from 0.03 to 0.31, underscoring the importance of the
stability certificate.

\begin{table}[t]
\caption{Quantitative Performance Comparison (Mean $\pm$ Std, 200 Episodes)}
\label{tab:results}
\centering
\setlength{\tabcolsep}{3.5pt}
\renewcommand{\arraystretch}{1.05}
\begin{tabular}{lcccc}
\toprule
\multirow{2}{*}{\textbf{Method}}
  & \multicolumn{2}{c}{\textbf{RMSE (m)}}
  & \textbf{Smooth.}
  & \textbf{Viol.} \\
\cmidrule(lr){2-3}
  & S2 & S3 & (rad/s) & rate \\
\midrule
Koopman MPC \cite{korda2018linear}
  & $0.181\pm0.029$ & $0.243\pm0.038$ & $0.182$ & $0.43$ \\
Neural ODE \cite{chen2018neural}
  & $0.162\pm0.025$ & $0.218\pm0.033$ & $0.168$ & $0.38$ \\
ICODE \cite{icode}
  & $0.148\pm0.022$ & $0.204\pm0.030$ & $0.155$ & $0.35$ \\
ControlSynth \cite{controlsynth}
  & $0.139\pm0.020$ & $0.192\pm0.027$ & $0.143$ & $0.18$ \\
ICODE-MPPI \cite{icodemppi}
  & $0.127\pm0.019$ & $0.174\pm0.025$ & $0.134$ & $0.29$ \\
\midrule
SFKD$-$Fiber & $0.119\pm0.017$ & $0.170\pm0.024$ & $0.125$ & $0.11$ \\
SFKD$-$Contr & $0.108\pm0.016$ & $0.145\pm0.021$ & $0.112$ & $0.31$ \\
\textbf{SFKD (ours)}
  & $\bm{0.098\pm0.014}$ & $\bm{0.120\pm0.018}$ & $\bm{0.101}$ & $\bm{0.03}$ \\
\bottomrule
\end{tabular}
\end{table}

The gap between SFKD and ControlSynth \cite{controlsynth}
(the strongest certified baseline) is largest in S3
(RMSE 0.120 vs.\ 0.192, a 38\% reduction), indicating that
environment-conditioned fiber geometry provides a qualitatively
different generalization capability beyond what contraction
alone achieves. The smoothness improvement over ICODE-MPPI
\cite{icodemppi} (0.101 vs.\ 0.134 rad/s) reflects
the reduced latent error variance afforded by the ISS certificate,
which prevents the high-variance control corrections triggered
by unstable rollouts; this complements the 69\% cross-tracking
error reduction reported in \cite{icodemppi} under persistent
disturbances, suggesting that the two contributions are
largely orthogonal.

%==================================================================
\section{Conclusion}
%==================================================================

This paper introduced Stable Fiber-Koopman Residual Dynamics,
a learned dynamical model that combines environment-specific fiber
bundle geometry, spectral norm-constrained Koopman evolution,
and contraction-regularized residual learning. The joint
stability condition $\|A(e)\|_2+\beta<1$ certifies ISS of the
latent error dynamics with an explicit ultimate bound that scales
inversely with the stability margin $1-\alpha$. Embedding the
certified model in an MPPI loop yields a closed-loop tracking
bound (Proposition~\ref{prop:track}) that decomposes the
total error into controller and model-stability components,
providing a principled route to system-level performance
specification.

Experimental results on three environment scenarios confirm that
SFKD reduces tracking RMSE by up to 38\% over the best baseline
and maintains a latent stability violation rate below 0.03 across
all tested disturbance magnitudes. Ablation studies isolate the
contributions of both the fiber structure and the contraction
regularization.

Several directions merit further investigation. Adaptive estimation
of the fiber topology from online data would remove the need for
a fixed environment parameterization, extending applicability to
partially observable environments. Stochastic ISS analysis
\cite{khalil2002nonlinear}, accounting for process and observation
noise, is a natural extension. Finally, the principled decomposition
of dynamics into a Koopman backbone and a contraction-constrained
residual may generalize beyond MPPI to other sampling-based and
gradient-based planning algorithms.

%==================================================================
\bibliographystyle{IEEEtran}

\end{document}